\documentclass[onecolumn,journal]{IEEEtran}
\ifCLASSINFOpdf
\else
\fi

\usepackage[T1]{fontenc}
\usepackage[latin9]{inputenc}
\usepackage{color}
\usepackage{array}
\usepackage{float}
\usepackage{mathrsfs}
\usepackage{mathtools}
\usepackage{amsthm}
\usepackage{amstext}
\usepackage{amssymb}
\usepackage{stmaryrd}
\usepackage{graphicx}
\usepackage{pgfplots}

\usepackage{tikz-cd}

\makeatletter

\theoremstyle{plain}
\newtheorem{thm}{\protect\theoremname}
\theoremstyle{plain}

\theoremstyle{plain}
\newtheorem{cor}[thm]{\protect\corollaryname}
\theoremstyle{plain}
\newtheorem{lem}[thm]{\protect\lemmaname}
\theoremstyle{definition}
\newtheorem{example}[thm]{\protect\examplename}
\theoremstyle{definition}
\newtheorem{defn}[thm]{\protect\definitionname}
\theoremstyle{plain}
\newtheorem{rem}[thm]{\protect\remarkname}

\AtBeginDocument{
	
}

\makeatother

  \providecommand{\corollaryname}{Corollary}
  \providecommand{\examplename}{Example}
  \providecommand{\lemmaname}{Lemma}
  \providecommand{\propositionname}{Proposition}
  \providecommand{\theoremname}{Theorem}
  \providecommand{\definitionname}{Definition}
  \providecommand{\remarkname}{Remark}

\newcommand{\Tr}{\operatorname{Tr}}

\newcommand{\floor}[1]{\left\lfloor {#1} \right\rfloor}
\newcommand{\Bk}[1]{{\Big( {#1} \Big) }}

\numberwithin{equation}{section}

\begin{document}

\title{Complete Weight Distribution and MacWilliams Identities for Asymmetric Quantum Codes}

\author{Chuangqiang Hu, Shudi Yang and Stephen S.-T. Yau     
\thanks{C. Hu's affiliation is Yau Mathematical Sciences Center, Tsinghua University, Beijing 100084, P. R. China.\protect\\
	S. Yang's affiliation is School of Mathematical
	Sciences, Qufu Normal University, Shandong 273165, P.R.China. 	\protect\\
	S. S.-T. Yau's affiliation is Department of Mathematical Sciences, Tsinghua University, Beijing 100084, P. R. China. 	\protect\\

	E-mail: huchq@mail2.sysu.edu.cn, yangshudi7902@126.com, yau@uic.edu }
\thanks{Manuscript received *********; revised ********.}
}

\maketitle

\begin{abstract}
In 1997, Shor and Laflamme defined the weight enumerators for quantum error-correcting codes and derived a MacWilliams identity.  We extend their work by introducing our double weight enumerators and complete weight enumerators. The MacWilliams identities for these enumerators can be obtained similarly. With the help of MacWilliams identities, we obtain various bounds for asymmetric quantum codes.
\end{abstract}


\IEEEpeerreviewmaketitle

\section{Introduction}

Quantum information theory is rapidly becoming a well-established
discipline. It shares many of the concepts of classical information
theory but involves new subtleties arising from the nature of quantum
mechanics. Among the central concepts in common between classical
and quantum information is that of error correction. Quantum error-correcting
codes have been initially discovered by Shor \cite{Shor1995} and Steane \cite{Steane1996,Steane1996-2}
in 1995-1996 for the purpose of protecting quantum information from noise in computation or communication \cite{Ekert1996,Knill1997,Nielsen2000,Calderbank1996,Ashikhmin2000,Ashikhmin2000_2}.
 This discovery of \cite{Shor1995} has revolutionized the field of quantum information and leads to a new research line.
 In \cite{Durdevich2000,Zanardi1997,Zanardi1997-2} noiseless quantum codes were built using group theoretic methods \cite{Xu2015,Zhang2017}.
In \cite{Bennett1996,Knill1997} quantum error correction was used to broader analyses of the physical principles. The authors in \cite{Calderbank1998,FengXing2008,WangXing2010} gave various new constructions of quantum error-correcting codes.

It is well known that if further information about the error process is available, more efficient codes can be designed. Indeed, in many physical systems, the noise is likely to be unbalanced between amplitude (X-type) errors and phase (Z-type) errors. Recently a lot of attention has been put into designing codes for this situation and into studying their fault tolerance properties \cite{Ioffe2007,WangXing2010,Sarvepalli2008,Fletcher2008Channel}. All these results use error models described by Kraus operators \cite{Kraus1983} that generalize Pauli operators.

  In the classical theory the famous MacWilliams identity gives a relationship between the weight distributions of a code $ C $ and its dual code $ C^{\perp} $ without knowing specifically the
  codewords of $ C^{\perp} $ or anything else about its structure \cite{MacW1963,MacWilliams1977}.
  The same technique was adapted to the quantum case by Shor and Laflamme \cite{ShorLa1997} generalizing the classical case and they derived a MacWilliams identity. Rains \cite{Rain1998} investigated the properties of quantum enumerators. In \cite{Rain1999}, Rains extended the work of \cite{ShorLa1997} to general codes by introducing quantum shadow enumerators.

 Several bounds are known for classical error-correcting codes. Delsarte \cite{Delsarte1973} investigated the Singleton and Hamming bounds using linear programming approach. The first linear programming bound was generalized by Aaltonen \cite{Aaltonen1979} to the nonbinary case. See \cite{Aaltonen1990,Laihonen1998,Levenshtein1995} for more information on available bounds for non-binary codes. Recently there has been intensive activity in the area of quantum codes. In particular, Knill and Laflamme \cite{Knill1997A} introduced the notion of the minimum distance of a quantum error-correcting code and showed
 that the error for entangled states is bounded linearly by the error for
 pure states. Shor and Laflamme \cite{ShorLa1997} presented a
 linear-programming bound for quantum error-correcting codes. Cleve \cite{Cleve1997Quantum} demonstrated connections between quantum stabilizer codes and classical codes and gave upper bounds on the best asymptotic capacity. Rains \cite{Rain1998,Rain1999} showed that the minimum distance of a quantum code is determined by its enumerators. Ashikhmin and Litsyn \cite{Ashikhmin1999} attained upper asymptotic bounds on the size of quantum codes.
 Aly \cite{Aly2008} established asymmetric Singleton
 and Hamming bounds on asymmetric quantum and subsystem
 code parameters. Sarvepalli \emph{et al.} \cite{Sarvepalli2009} studied asymmetric quantum codes and derived upper bounds on the code parameters using linear programming.
 Wang \emph{et al.} \cite{WangXing2010} extended the characterization
 of non-additive symmetric quantum codes given in \cite{FengXing2006,FengXing2008} to the asymmetric case and obtained an asymptotic bound from algebraic geometry codes.

It should be mentioned that there is another weight enumerator for a classical code that contain more detailed information about
the codewords. Namely, the complete weight enumerator, which enumerates the codewords according to the number of alphabets of each kind contained in each codeword. MacWilliams \cite{MacW1963,MacWilliams1977} also
proved that there is an identity between the complete weight enumerators of $ C $ and its dual code $ C^{\perp} $. The complete weight enumerator and weight enumerator of classical codes have been studied extensively, see \cite{vega2012weight,ding2009weight,liYue2014weight,yang2015weightenu,Yang2017constru,Yang2016complete,yang2015weight} and the reference therein.
However, to the best of our knowledge, there is no quantum analog complete weight enumerators as in classical coding theory. Therefore the purpose of the present paper is to introduce the notions of double weight enumerator and complete weight enumerator and then generalize the MacWilliams identities about complete weight enumerators from classical coding theory to the quantum case.
Using the generalized MacWilliams identities, we will find new upper bounds on the minimum distance of asymmetric quantum codes.

Here is the plan of the rest of this paper. In Section \ref{sec:sym}, we introduce some basic definitions and notations on symmetric and asymmetric quantum codes. In Section \ref{sec:enu}, we establish our main result on quantum MacWilliams identities by defining double weight enumerators and complete weight enumerators of quantum codes. In Section \ref{sec:Kraw}, we give a short survey of properties of the Krawtchouk polynomials and we prove the key inequality that allows us to get new upper bounds of the minimum distance of an asymmetric quantum codes.
In Section \ref{sec:upper}, we apply the key inequality to obtain Singleton-type, Hamming-type and the first linear-programming-type bounds for asymmetric quantum codes.

\section{symmetric and asymmetric quantum codes}\label{sec:sym}
We begin with some basic definitions and notations. Let $ \mathbb{C} $ be a complex number field. We regard $ \mathbb{C}^2 $ as a Hilbert space with orthonormal basis $ |0 \rangle $ and $ | 1 \rangle  $. Denote by $ (\mathbb{C}^2)^{\otimes n } = \mathbb{C}^{2^n } $ the $ n $-th tensor of $ \mathbb{C}^2 $. This space enables us to transmit $ n $ bits of information. The coordinate basis is given by
\[ |j \rangle = | j_0 \rangle \otimes | j_1 \rangle  \otimes \ldots \otimes | j_{n-1} \rangle , \text{for each } j_r \in \{0, 1\}. \]
For two quantum states $ | u \rangle $ and $ |v \rangle $ in $ \mathbb{C}^{2^n} $ with
\[|u\rangle = \sum_{j} u_{j} | j \rangle, \quad  |v\rangle = \sum_{j} v_{j} | j \rangle,  \]
 the Hermitian inner product of $ | u\rangle $ and $ | v\rangle
 $ is defined by \[ \langle u | v \rangle = \sum_{j} \bar {u}_j v_j. \]

In the process of transmission over a channel the information can be altered by errors. There are several models of channels. Perhaps the most popular one is the completely depolarized channel, among which a vector $ v \in \mathbb{C}^2 $ can be altered by one of the following error operators:
\[
\sigma_{x}= \left[\begin{matrix*}[c]
0 &  1 \\
1 & 0
\end{matrix*}\right],
\sigma_{y}=\left[\begin{matrix*}[c]
0 &  -i \\
i & 0
\end{matrix*}\right]  ,
\sigma_{z}=\left[\begin{matrix*}[c]
1 &  0\\
0 & -1
\end{matrix*}\right]  .
\]
The error operators action on $ \mathbb{C}^{2} $ constitute a set
\[  E := \{e = \sigma_{0} \otimes \ldots \otimes \sigma_{n-1} | \sigma_r  \in \{ I, \sigma_{x} , \sigma_{y} , \sigma_{z} \}  \} . \]
For $ e \in E $, the number of non-identity matrices in expression of $ e $ is called the weight of $ e $ which is denoted by $  w_{Q} (e) $. Similarly, we denote by $ N_{x}(e)  $, $ N_{y}(e) $ and $ N_{z}(e)  $ the number of the matrices $ \sigma_{x} $, $ \sigma_{y} $ and $ \sigma_{z} $ occurred in the expression of $ e $ respectively. It is clear that
 \[w_{Q} (e) = N_{x}(e) + N_{y}(e)+N_{z}(e) . \]
 It is well known that each $ e \in E $ is the composition of two kinds of error operators, i.e. the bit flip and the phase flip. Precisely,  fix an error operator $ e $, there exist vectors $ a = ( a_0 \ldots a_{n-1}) \in \mathbb{F}_2^ n $ and $ b = ( b_0 \ldots b_{n-1})\in \mathbb{F}_2^ n $ such that
\begin{equation}\label{eq:e=XZ}
 e = i ^{a \cdot b } X(a) Z(b),
\end{equation}
where
\[ X(a) = \omega_{0} \otimes \ldots\otimes \omega_{n-1}, \quad  Z(b) = \omega_{0}^{\prime} \otimes \ldots\otimes \omega_{n-1}^{\prime} , \]
and
\[\omega_{i} = \begin{cases}
	I, & \text{ if }a_i = 0 ,\\
	\sigma_{x} &\text{ if } a_i = 1,
\end{cases}
\quad
\omega_{i}^{\prime} = \begin{cases}
I, &\text{ if } b_i = 0 ,\\
\sigma_{z} &\text{ if } b_i = 1.
\end{cases}
\]
We define the $ X $-weight $ w_{X}(e)  $ and the $ Z $-weight $ w_{Z}(e)  $ to be the Hamming weights of $ a $ and $ b $ of Equation (\ref{eq:e=XZ}) respectively. In fact, they alternatively can be defined as
\[w_{X}(e) = N_{x}(e) + N_{y}(e), \quad w_{Z}(e) = N_{y}(e) + N_{z}(e). \]
In the following section we want to investigate some natural partitions of the set $ E $, so we define
\begin{align*}
	E[i,j,k] &:= \{ e \in E |  N_{x}(e)= i ,   N_{y}(e) = j  ,  N_{z}(e)= k \},\\
 E[i,j] &:= \{ e \in E |  w_{X}(e)= i ,  w_{Z}(e) = j  \},\\
E[i] &:= \{ e \in E |   w_{Q}(e)=i  \}.
\end{align*}
\begin{defn}
 A quantum code of length $ n $ is a linear subspace of $ \mathbb{C}^{2^n} $ with dimension $K \geqslant 1$. Such a quantum code can be denoted as $ ((n,K)) $ code or $ [[n,k]] $ code, where $ k = \log K $.
\end{defn}
\begin{rem}
	Here and thereafter, the logarithms are base 2.
\end{rem}
\begin{defn}
	Let $ Q $ be a quantum code. An error $ e $ in $ E $ is called detectable if
	\[ \langle v |e | w \rangle = 0  \]
	for all orthogonal vectors $ v $ and $ w $ from the code $ Q $.
\end{defn}

 Denote by $ P$ the orthogonal projection from $ C^{2^n } $ onto a quantum code $Q $. We have an alternative definition for detectable errors. It is deduced from \cite{Ashikhmin1999} that $ e $ is detectable if and only if
\[ Pe P = \lambda_e P  \]
for a constant $ \lambda_e $ depending on $ e $.

\begin{defn}
		Let $ Q $ be a quantum code with parameters $ ((n,K))$. The minimum distance of $Q $ is the maximum integer $ d $ such that  any error $ e  \in E [ i ] $ with $ i  < d $ is detectable. Such a quantum code is called a symmetric quantum code with parameters $ ((n,K, d )) $ or $ [[n,k,d]] $. If the integers $ d_x $ and $ d_z $ are the maximum integers such that each error $ e  \in E [ i ,j ] $ with $ i  < d_x , j < d_z $ is detectable, then we call $ Q $ an asymmetric quantum code with parameters $ ((n,K, d_z/ d_x ))$ or $ [[n,k,d_z/ d_x]] $.
\end{defn}
The classical Singleton bound can be extended to quantum codes.
\begin{thm}
Let $ Q $ be a quantum code with parameters $ [[n,k,d]]$. We have
\[  n \geqslant k + 2d -2. \]
\end{thm}
In \cite{Aly2008,WangXing2010}, the authors have proved the Singleton bound for stabilizer asymmetric quantum codes. That is
\[
  n \geqslant k + d_x+ d_z -2.
\]
However, we can not find the proof of Singleton bound for general quantum codes.

\section{weight distributions and enumerators}\label{sec:enu}
The weight distribution for classical codes can be generalized to the case of quantum codes. According to \cite{ShorLa1997},
the weight distribution for quantum codes is defined by the following two sequences of numbers
\begin{align*}
B_i& = \frac{1}{K^2} \sum_{e \in E[i] } \Tr^2(e P ), \\
B_i^{\bot} & = \frac{1}{K} \sum_{e \in E[i] } \Tr  (e P e P ).
\end{align*}
Moreover, the corresponding weight enumerator is defined to be the following two bivariate polynomials
\begin{align*}
 B(X, Y) &:= \sum_{i=0}^{n} B_i X^{n-i} Y^i ,\\
 B^{\bot}(X, Y) &:= \sum_{i=0}^{n} B_i^{\bot} X^{n-i} Y^i .
 \end{align*}
In a similar manner, we introduce the double weight distribution
\begin{align*}
C_{i,j} &= \frac{1}{K^2} \sum_{e \in E[i,j] } \Tr^2(e P ), \\
C_{i,j}^{\bot}& = \frac{1}{K} \sum_{e \in E[i,j] } \Tr (e P e P),
\end{align*}
and the complete weight distribution
 \begin{align*}
 D_{i,j,k} &= \frac{1}{K^2} \sum_{e \in E[i,j,k] } \Tr^2(e P ), \\
 D_{i,j,k}^{\bot} &= \frac{1}{K} \sum_{e \in E[i,j,k] } \Tr (e P e P ).
 \end{align*}
Then the double weight enumerator and the complete weight enumerator are defined by
\begin{align*}
  C(X,Y,Z,W)& := \sum_{i,j = 0}^{n}   C_{i,j }X^{n - i} Y^ i  Z^{n - j} W^ j ,\\
  C^{\bot}(X,Y,Z,W) &:= \sum_{i,j = 0}^{n}   C_{i,j }^{\bot} X^{n - i} Y^ i  Z^{n - j} W^ j,\\
  D(X,Y,Z,W) &:= \sum_{  i+j+k \leqslant n }  D_{i,j,k} X^i Y^j Z^k W^{n - i - j - k}, \\
 D^{\bot}(X,Y,Z,W) &:= \sum_{  i+j+k \leqslant n }  D_{i,j,k}^{\bot} X^i Y^j Z^k W^{n - i - j - k}  .
 \end{align*}

These enumerators are related by the following theorem.
\begin{thm}\label{thm:BandD}
	 Let $ Q $ be a quantum code with enumerators $ B,B^{\bot},C,C^{\bot},D$ and $D^{\bot} $. Then the following four identities hold:
 \begin{equation}\label{eq:B=D}
 	B(X,Y) = D( Y,Y,Y,X),
 \end{equation}
  \begin{equation}\label{eq:B2=D2}
B^{\bot}(X,Y) = D^{\bot}( Y,Y,Y,X),
 \end{equation}
  \begin{equation}\label{eq:C=D}
C(X,Y,Z,W) = D(YZ,YW,XW,XZ),
\end{equation}
  \begin{equation}\label{eq:C2=D2}
 C^{\bot}(X,Y,Z, W) = D^{\bot}(YZ,YW,XW,XZ).
\end{equation}
\end{thm}

\begin{proof}
	We shall investigate the explicit expressions for these enumerators. Recall that
the coordinate basis for $ \mathbb{C}^{2^n } $ is given by
\[  |j \rangle = |j _0 \rangle \otimes \ldots \otimes |j _{n-1}  \rangle ,\]
where $ 0 \leqslant j < 2^n$ and $j = j_0 + j_1 2+\ldots + j_{n-1} 2^{n-1}$.
Denote by $ e_{i,j}, p_{i,j} $ the entries of $ e $ and $ P $ respectively with respect to our coordinate basis.
For  $e = \sigma_{0} \otimes \ldots \otimes \sigma_{n-1} \in E $,  the identity $  e |i\rangle = \sum _{j } e_{i,j}| j\rangle $ implies
\[ e_{i,j} = (\sigma_{0})_{i_0 , j_0}(\sigma_{1})_{i_1 , j_1}\ldots (\sigma_{n-1})_{i_{n-1}, j_{n-1}} .\]
According to the definition of $D$, we get
\begin{align*}
  D(X,Y,Z,W) &=  \sum_{i,j,k} X^i Y^j Z^k W^{n-i-j-k} D_{i,j,k}\\
  & = \frac{1}{K^2} \sum_{i,j,k} X^i Y^j Z^k W^{n-i-j-k}\sum_{e \in E[i,j,k]} \sum_{r,s,t,u} e_{r,s} p_{s,r} e_{t,u} p_{u,t} \\
   & = \frac{1}{K^2} \sum_{r,s,t,u}p_{s,r}  p_{u,t} \sum_{e \in E }    X^{N_x(e)} Y^{N_y(e)} Z^{N_z(e)} W^{n-w_{Q}(e)}e_{r,s} e_{t,u}  \\
   &= \frac{1}{K^2} \sum_{r,s,t,u}p_{s,r}  p_{u,t} \prod_{\lambda=0}^{n-1} d_{\lambda}(X,Y,Z,W),
\end{align*}
where \[ d_\lambda(X,Y,Z,W) =(\sigma_{x})_{r_\lambda,s_\lambda}(\sigma_{x})_{t_\lambda,u_\lambda}X+
(\sigma_{y})_{r_\lambda,s_\lambda}(\sigma_{y})_{t_\lambda,u_\lambda}Y+
(\sigma_{z})_{r_\lambda,s_\lambda}(\sigma_{z})_{t_\lambda,u_\lambda}Z+ (I)_{r_\lambda,s_\lambda}(I)_{t_\lambda,u_\lambda}W . \]
By the same method, we have
\begin{align*}
C(X,Y,Z,W) &=  \sum_{i,j} X^{n-i} Y^i Z^{n-j} W^{j} C_{i,j}\\
& =  \frac{1}{K^2} \sum_{i,j} X^{n-i} Y^i Z^{n-j} W^{j}  \sum_{e \in E[i,j]} \sum_{r,s,t,u} e_{r,s} p_{s,r} e_{t,u} p_{u,t} \\
& = \frac{1}{K^2} \sum_{r,s,t,u}p_{s,r}  p_{u,t} \sum_{e \in E }    X^{n-w_X(e)} Y^{w_X(e)} Z^{n-w_Z(e)} W^{w_{Z}(e)}e_{r,s} e_{t,u}  \\
&= \frac{1}{K^2} \sum_{r,s,t,u}p_{s,r}  p_{u,t} \prod_{\lambda=0}^{n-1} c_{\lambda}(X,Y,Z,W),
\end{align*}
where \[ c_\lambda(X,Y,Z,W) =(\sigma_{x})_{r_\lambda,s_\lambda}(\sigma_{x})_{t_\lambda,u_\lambda}YZ+
(\sigma_{y})_{r_\lambda,s_\lambda}(\sigma_{y})_{t_\lambda,u_\lambda}YW+
(\sigma_{z})_{r_\lambda,s_\lambda}(\sigma_{z})_{t_\lambda,u_\lambda}XW+ (I)_{r_\lambda,s_\lambda}(I)_{t_\lambda,u_\lambda}XZ . \]

Moreover, we obtain
\begin{align*}
B(X,Y) &=  \sum_{i,j} X^{n-i} Y^i B_{i}\\
& =  \frac{1}{K^2} \sum_{i,j} X^{n-i} Y^i   \sum_{e \in E[i]} \sum_{r,s,t,u} e_{r,s} p_{s,r} e_{t,u} p_{u,t} \\
& = \frac{1}{K^2} \sum_{r,s,t,u}p_{s,r}  p_{u,t} \sum_{e \in E }    X^{n-w_Q(e)} Y^{w_Q(e)}  e_{r,s} e_{t,u}  \\
&= \frac{1}{K^2} \sum_{r,s,t,u}p_{s,r}  p_{u,t} \prod_{\lambda=0}^{n-1} b_{\lambda}(X,Y),
\end{align*}
where \[ b_\lambda(X,Y) =(\sigma_{x})_{r_\lambda,s_\lambda}(\sigma_{x})_{t_\lambda,u_\lambda}Y+
(\sigma_{y})_{r_\lambda,s_\lambda}(\sigma_{y})_{t_\lambda,u_\lambda}Y+
(\sigma_{z})_{r_\lambda,s_\lambda}(\sigma_{z})_{t_\lambda,u_\lambda}Y+ (I)_{r_\lambda,s_\lambda}(I)_{t_\lambda,u_\lambda}X . \]
So Equations (\ref{eq:B=D}) and (\ref{eq:C=D}) are deduced from
\[ d_{\lambda} (YZ, YW, XW, XZ) = c_{\lambda}(X,Y,Z,W),  \]
and
\[ d_{\lambda} (Y, Y, Y, X) = b_{\lambda}(X,Y).  \]
Applying the same argument, Equations (\ref{eq:B2=D2}) and (\ref{eq:C2=D2}) follow immediately.
\end{proof}
The classical MacWilliams identity provides a relationship between classical linear codes and their dual codes.
It is interesting to see that the MacWilliams identity also holds for quantum weight enumerators \cite{ShorLa1997}. That is
	\begin{align} \label{eq:MacB}
		B(X,Y) = \frac{1}{K} B^{\bot} (\frac{X+3 Y}{2},  \frac{X-Y}{2}) .
		\end{align}
			Our main result is to show that quaternary MacWilliams identities for the double weight enumerator and complete weight enumerator hold similarly.
\begin{thm}[quaternary MacWilliams identities]\label{thm:quaMac}
	 Let $ Q $ be an $ ((n,K)) $ quantum code. With the notation introduced above, we have
\begin{align}
  C(X,Y,Z,W) &= \frac{1}{K} C^{\bot}(Z+W, Z-W, \frac{X+Y}{2}, \frac{X-Y}{2}) ,\label{eq:C=D1}\\
  D(X,Y,Z,W)& = \frac{1}{K} D^{\bot}(\frac{X- Y- Z+W }{2}, \frac{ -X+ Y- Z+W }{2},\frac{-X- Y+ Z+W }{2}, \frac{X+ Y+ Z+W }{2} ) .  \label{eq:D=D1}
\end{align}		
\end{thm}

\begin{proof}
	We maintain all notations in the proof of the previous theorem. Using the same method, one can show that
	\[ D^{\bot}(X,Y,Z,W) = \frac{1}{K } \sum_{r,s,t,u}p_{s,r}  p_{u,t} \prod_{\lambda=0}^{n-1} d_{\lambda}^{\bot}(X,Y,Z,W),  \]
	where
\[	d_{\lambda}^{\bot}(X,Y,Z,W) =   (\sigma_{x})_{r_\lambda,u_\lambda}(\sigma_{x})_{t_\lambda,s_\lambda}X+
	(\sigma_{y})_{r_\lambda,u_\lambda}(\sigma_{y})_{t_\lambda,s_\lambda}Y+
	(\sigma_{z})_{r_\lambda,u_\lambda}(\sigma_{z})_{t_\lambda,s_\lambda}Z+ (I)_{r_\lambda,u_\lambda}(I)_{t_\lambda,s_\lambda}W . \]
	
	We can check directly that for arbitrary $ r_{\lambda}, s_{\lambda}, t_{\lambda}, u_{\lambda} \in \{0,1\} $,
	\[  d_{\lambda}(X,Y,Z,W) =   d^{\bot}_{\lambda}(\frac{X- Y- Z+W }{2}, \frac{ -X+ Y- Z+W }{2},\frac{-X- Y+ Z+W }{2}, \frac{X+ Y+ Z+W }{2} ) .  \]
	This implies Equation (\ref{eq:D=D1}).
	Using this equation and the relationship between $ C $ and $ D $, we get
\begin{align*}
	C(X,Y,Z,W)& = D(YZ, YW, XW, XZ)\\
	& = \frac{1}{K} D^{\bot}(\frac{(X+Y)(Z-W)}{2},\frac{(X-Y)(Z-W)}{2},\frac{(X-Y)(Z+W)}{2}, \frac{(X+Y)(Z+W)}{2}) \\
		& = \frac{1}{K} C^{\bot}(Z+W, Z-W, \frac{X+Y}{2},\frac{X-Y}{2} ),
\end{align*}
where in the last step we use Equation \eqref{eq:C2=D2}. This completes the proof of the theorem.
\end{proof}

The following theorem generalize Theorem 3 in \cite{Ashikhmin1999} to the case of double weight distribution.
\begin{thm}\label{thm:double weight}
	Let $ ((Q,K,d_z/ d_x)) $ be an asymmetric quantum code with double weight distribution $ C_{i,j} $ and $ C_{i,j}^\bot $. Then
	\begin{enumerate}
		\item $ C_{i,j}^{\bot} \geqslant C_{i,j} \geqslant 0 $ for  $ 0\leqslant i ,j \leqslant n $, and $ C_{0,0} =C_{0,0}^{\bot} = 1  $.
		\item If $ t_x $, $ t_z $ are the two largest integers such that $ C_{i,j} =C_{i,j}^{\bot} $ for  $ i< t_x $ and $ j < t_z $, then
		$ d_x = t_x $ and $ d_z= t_z$.
	\end{enumerate}
\end{thm}
\begin{proof}
	The proof is similar to that of Theorem 3 in \cite{Ashikhmin1999} and so is omitted here.
\end{proof}

\begin{rem}
	Suppose that $ Q $ is an additive quantum code of length $ n $ constructed from a classical linear code $ C $ over $ \mathbb{F}_4 $. Then the weight distribution of $ Q $ is nothing but the classical distribution induced by $  C  $ and its symplectic dual $ C_{th}^{\perp } $. That is
	\begin{align*}
	B_i &= \# \{c\in C : w_H(c)=i\},\\
	B_i^{\perp} &= \# \{c\in C_{th}^{\perp} : w_H(c)=i\},
	\end{align*}
	where $  w_H(c) $ denotes the Hamming weight of a codeword $ c $, see \cite{Ashikhmin2000,Rain1999}.
     Let $ \alpha $ be a fixed primitive element of $ \mathbb{F}_4 $, i.e., $ \mathbb{F}_4 =\{\alpha,\alpha^2,\alpha^3=\alpha+\alpha^2,0\}$. Write
		\begin{align*}
		c=(c_1,c_2,\cdots, c_{n})=(a_1 \alpha + b_1 \alpha^2,a_2 \alpha + b_2 \alpha^2,\cdots,a_n \alpha + b_n \alpha^2) \in
		\mathbb{F}_4^n,
		\end{align*}
		where $( a_1,a_2,\cdots,a_n,b_1,b_2,\cdots,b_n) \in \mathbb{F}_2^{2n} $. Then the double weight distribution of $ Q $ is deduced from $ C $ and its symplectic dual $ C_{th}^{\perp } $, i.e.,
	\begin{align*}
	C_{i,j} &= \# \{c=(c_1,c_2,\cdots, c_{n})\in C : \sum_{s=1}^{n }a_s=i,  \sum_{s=1}^{n}b_s=j  \},\\
		C_{i,j}^{\perp} &= \# \{c=(c_1,c_2,\cdots, c_{n})\in C_{th}^{\perp} : \sum_{s=1}^{n }a_s=i,  \sum_{s=1}^{n}b_s=j  \}.
	\end{align*}
 	For a vector $ c = (c_1, c_2, \cdots, c_{n})\in \mathbb{F}_4^n$, the composition of $ c $, denoted by $ \textup{comp}(c) $, is defined as
	$  \textup{comp}(c)=(k_1,k_2,k_3,n-k_1-k_2-k_3) $,
	 where $ k_j $ $ (j\neq 0) $ is the number of components $ c_s $ of $ c $ that are equal to $ \alpha^j $. Then the complete weight distribution of $ Q $ is also deduced from $ C $ and its symplectic dual $ C_{th}^{\perp } $, i.e.,
	 \begin{align*}
	 D_{i,j,k} &= \# \{c\in C :  \textup{comp}(c)=(i,j,k,n-i-j-k)\}, \\
	 D_{i,j,k}^{\perp} &= \# \{c \in C_{th}^{\perp} :\textup{comp}(c)=(i,j,k,n-i-j-k)\}.
	 \end{align*}
	\end{rem}

Here we provide a concrete example to illustrate our main results.
 \begin{example}
  Let $ (q,m)=(4,2) $. A $ [5, 3, 3] $ Hamming code $ \mathcal{H}_2 $ over $ \mathbb{F}_4 $ of length $ n=(q^m-1)/(q-1)=5 $ has check matrix
     \begin{align*}
     \left[\begin{matrix*}[c]
       1 &  0 &  1 & \alpha^2 & \alpha^2 \\
     0  & 1 & \alpha^2 & \alpha^2  & 1
     \end{matrix*}\right],
    \end{align*}
    where $ \alpha $ is a fixed primitive element of $ \mathbb{F}_4 $. Its dual code $ \mathcal{H}_2^{\perp}  $ is a $  [5, 2, 4]  $ linear code over $ \mathbb{F}_4 $. The code $ \mathcal{H}_2^{\perp}  $ induces an additive quantum code $ Q $ with parameters $ [[n,n-2m,3]]=[[5,1,3]] $. The weight enumerator of $ Q $ is computed from $C:= \mathcal{H}_2^{\perp}  $ and its symplectic dual $  C_{th}^{\perp } :=\mathcal{H}_2 $, namely,
    \begin{align*}
     B(X,Y)&=X^5+15XY^4,\\ 
    B^{\perp} (X,Y)&=X^5 + 30X^2Y^3 + 15XY^4 + 18Y^5.
   \end{align*}
   One verifies that the MacWilliams identity \eqref{eq:MacB} holds for $ B $ and $  B^{\perp}  $.
   The double weight enumerator of $ Q $ is as follows
   \begin{align*} 
   C (X,Y,Z,W)&=X^5 Z^5 + 5X^3Y^2 Z^3 W^2 + 5X^3Y^2 Z  W^4 + 5X Y^4 Z^3 W^2,\\
      C^{\perp} (X,Y,Z,W)& = X^5 Z^5 + X^5  W^5 
   + 5X^4 Y  Z^3  W^2 + 5X^4 Y  Z^2  W^3 
   + 5X^3 Y^2  Z^4  W + 5X^3 Y^2  Z^3  W^2\\ 
   &  \phantom{=} + 5X^3 Y^2  Z^2  W^3 + 5X^3 Y^2  Z   W^4  
  + 5X^2 Y^3  Z^4  W + 5X^2 Y^3  Z^3  W^2 
   + 5X^2 Y^3  Z^2  W^3\\
  &  \phantom{=} + 5X^2 Y^3  Z   W^4  
   + 5X  Y^4  Z^3  W^2+ 5X Y^4  Z^2  W^3   
   +    Y^5  Z^5  +   Y^5    W^5 .  
   \end{align*}
   The complete weight enumerator of $ Q $ is given below
   \begin{align*}
    D(X,Y,Z,W)&= W^5 + 5 W X^2 Y^2 + 5 W X^2 Z^2 + 5 W Y^2 Z^2,\\
    D^{\perp}(X,Y,Z,W)& =
    W^5 + X^5+ Y^5 + Z^5+ 5 W^2 X^2 Y + 5 W^2 X^2 Z + 5 W^2 X Y^2 + 5 W^2 X Z^2
    + 5 W^2 Y^2 Z \\
    &  \phantom{=} + 5 W^2 Y Z^2
     + 5 W X^2 Y^2 + 5 W X^2 Z^2 + 5 W Y^2 Z^2
    + 5 X^2 Y^2 Z + 5 X^2 Y Z^2 + 5 X Y^2 Z^2  .
   \end{align*}
   The above experimental results by Magma are consistent with the conclusions of Theorems \ref{thm:BandD} and \ref{thm:quaMac}.
\end{example}

\section{Krawtchouk polynomials and the key inequality} \label{sec:Kraw}
In this section, we introduce Krawtchouk polynomials and summarize their properties. This allows us to obtain the close relationship between $ C_{i,j} $ and $ C_{i,j}^\bot $ of an asymmetric quantum code. Then we propose the key inequality which enables us to reduce the problem of upper-bounding the size of asymmetric  quantum codes to a problem of finding polynomials possessing special properties.

Fix an integer $ n $. For $ 0 \leqslant i \leqslant n  $, the polynomial
\[ P_i (x) = \sum_{ j = 0 }^{i } (-1)^j
\binom{x}{j}\binom{n- x}{i-j} \]
is called the $ i $-th Krawtchouk polynomial. The first few polynomials are
\begin{align*}
P_0 (x)= 1 ,  P_1(x)= n -2x , P_2 (x ) = 2 x^2 - 2n x + \binom{n}{2}, \ldots
\end{align*}
These polynomials have the generating function
\[ (X+Y)^{n-r}(X-Y)^r = \sum_{i = 0}^n P_i (r)  X^{n-i}Y^i .   \]

Now we recall several important properties of the Krawtchouk polynomials, see \cite{MacWilliams1977} for more information.

The Krawtchouk polynomials satisfy the reciprocity formula
\begin{align}\label{eq:sym}
\binom{n}{i}P_s(i)=\binom{n}{s}P_i(s).
\end{align}

They also have the following property
\begin{align*}
\sum_{ i = 0 }^n \binom{n-i}{n-j}P_i(x)=2^j \binom{n-x}{j}.
\end{align*}

The Krawtchouk polynomials are orthogonal to each other, i.e.,
\begin{align}\label{eq:orth}
\sum_{ i = 0 }^n  P_r(i) P_i(s) = 2^n \delta_{r,s}.
\end{align}

Many important facts follow from this orthogonality. For example, there is a three-term recurrence:
\begin{align}\label{eq:recur}
(i+1) P_{i+1}(x)= (n-2x) P_i(x)-(n-i+1)P_{i-1}(x).
\end{align}
The Christoffel-Darboux formula (see Corollary 3.5 of \cite{Levenshtein1995}) also holds:
\begin{align}\label{eq:CDeq}
   P_{t+1}(x)P_t(a)- P_t(x)P_{t+1}(a)
   = \dfrac{2(a-x)}{t+1} \binom{n}{t}
     \sum_{i = 0}^t \dfrac{P_i(x)P_i(a)}{\binom{n}{i}}.
\end{align}

Using \eqref{eq:sym} and \eqref{eq:recur}, the ratio $ P_t(x+1)/P_t(x) $ is given by McEliece $ et.al.$ \cite{McEliece1977}
\begin{align}\label{eq:ratio}
\dfrac{P_t(x+1)}{P_t(x)}=\dfrac{n-2t+\sqrt{(n-2t)^2-4j(n-x)}}{2(n-x)}
\Big(1+O\Bk{\frac{1}{n}} \Big).
\end{align}

We will also need a result on asymptotic behaviour of the smallest root $ r_t $ of $ P_t(x) $. For $ t $ growing linearly in $ n $ and $ \tau=t/n $ (see, e.g., Eq A.20 of \cite{McEliece1977})
\begin{align}\label{eq:bound1}
\gamma_t =\dfrac{r_t}{n}= \dfrac{1}{2}-\sqrt{\tau(1-\tau)}+ o(1).
\end{align}
Remember that $ o(1) $ tends to $ 0 $ as $ n $ grows.

The following theorem shows that the Krawtchouk polynomials have close relation to double weight distribution $ C_{i,j} $ and $ C_{i,j}^\bot $ of a quantum code.
\begin{thm}\label{thm:Cij}
	Let $ ((Q,K,d_z/ d_x)) $ be an asymmetric quantum code with double weight distribution $ C_{i,j} $ and $ C_{i,j}^\bot $. Then
	\begin{align}
		C_{i,j} &= \frac{1}{2^n K} \sum_{r,s = 0}^{n } P_i(s) P_j (r) C_{r,s}^{\bot},\label{eq:Cij}\\
	C_{r,s}^{\bot} &= \frac{K }{2^n} \sum_{i,j = 0}^{n } P_r(j) P_s(i) C_{i,j} .\label{eq:Cij2}
	\end{align}
\end{thm}
\begin{proof}
	We have from Theorem \ref{thm:quaMac} that
	\begin{align*}
	& \sum_{i,j = 0}^{n}	C_{i,j} X^{n-i}Y^i Z^{n-j}W^ j  \\
	&=  \frac{1}{2^n  K} \sum_{r,s = 0}^{n}   C_{r,s}^{\bot} (Z+W)^{n-r}( Z-W)^r ( X+Y )^{n-s }( X-Y )^s \\
	&=  \frac{1}{2^n  K} \sum_{r,s = 0}^{n}   C_{r,s}^{\bot} \sum_{i,j=0}^n P_i (s) P_j(r) X^{n -i } Y^ i Z^{n-j}W^j   \\
	&=  \frac{1}{2^n  K} \sum_{i,j=0}^n\sum_{r,s = 0}^{n}   C_{r,s}^{\bot} P_i (s) P_j(r)  X^{n -i } Y^ i Z^{n-j}W^j .
	\end{align*}
	 This completes the proof of \eqref{eq:Cij}. Taking into account that
	\begin{align*}
	  C^{\bot} (X,Y,Z,W) &= K\cdot C(Z+W, Z-W, \frac{X+Y}{2}, \frac{X-Y}{2})
	\end{align*}
	from \eqref{eq:C=D1}, the identity \eqref{eq:Cij2} follows immediately.
\end{proof}

The key inequality is given below, which will be needed in the sequel.

\begin{lem}\label{lem:K}
	Let $ Q $ be an $ ((n, K, d_z/ d_x)) $ quantum code. Assume that the polynomial
	$ f(x,y) = \sum_{i,j=0}^n \alpha_{i,j} P_i(y) P_j(x) $
	satisfies the following conditions
	\begin{enumerate}
		\item $ \alpha_{i,j} \geqslant 0 $ for $0   \leqslant i,j \leqslant n $,
		\item  $ f(r,s) > 0 $  for $0   \leqslant r  < d_x $ and $ 0   \leqslant s  < d_z $,
		\item  $ f(r,s) \leqslant 0 $  for $ r  \geqslant d_x $ or $  s  \geqslant d_z $.
	\end{enumerate}
Then
 \[ K \leqslant \frac{1}{2^n} \max_{0 \leqslant i < d_x ,  0 \leqslant j < d_z} \frac{f(i,j)}{\alpha_{i,j}}  . \]
\end{lem}

\begin{proof}
	It follows from Theorems \ref{thm:double weight} and \ref{thm:Cij} that
\begin{align*}
	2^n K \sum_{i=0}^{d_x-1} \sum_{j=0}^{d_z-1} \alpha_{i,j} C_{i,j}
	& \leqslant 	2^n K \sum_{i,j = 0}^n \alpha_{i,j} C_{i,j} \\
	&=    \sum_{i,j = 0}^n \alpha_{i,j} \sum_{r,s=0}^n P_i(s) P_j (r) C_{r,s}^{\bot}   \\
	&= \sum_{r,s=0}^n f(r,s) C_{r,s}^{\bot} \\
	& \leqslant \sum_{r=0}^{d_x-1} \sum_{s=0}^{d_z-1} f(r,s) C_{r,s}^{\bot} \\
		& = \sum_{r=0}^{d_x-1} \sum_{s=0}^{d_z-1} f(r,s) C_{r,s} .
\end{align*}
Thus we have
\[ 2^n K \leqslant \frac{\sum_{i=0}^{d_x-1} \sum_{j=0}^{d_z-1} f(i,j) C_{i,j}}{\sum_{i=0}^{d_x-1} \sum_{j=0}^{d_z-1} \alpha_{i,j} C_{i,j} } \leqslant  \max_{0 \leqslant i < d_x ,  0 \leqslant j < d_z} \frac{f(i,j)}{\alpha_{i,j}},  \]
completing the proof of this lemma.
\end{proof}

The following lemma is a consequence of Lemma \ref{lem:K}.
\begin{lem}\label{lem:K2}
	Let $ Q $ be an $ ((n, K, d_z/ d_x)) $ quantum code. Define	$ f(x,y) = f_1(x)f_2(y)$ where  $f_1(x)=\sum_{ j=0}^n \beta_{ j}   P_j(x) $ and
	$f_2(y)=\sum_{ i=0}^n \alpha_{ i}   P_i(y) $. Assume that the polynomial
    $ f(x,y) $
	satisfies the following conditions
	\begin{enumerate}
		\item  For even $ i $, $ \alpha_i>0 $, $ \beta_i>0 $. For odd $ i $, $ \alpha_{i }  = \beta_{i } = 0 $.
		\item  $ f(r,s) \geqslant 0 $  for $0   \leqslant r  < d_x $ and $ 0   \leqslant s  < d_z $. For even $ r $, $ f_1(r) >0 $, $ f_2(r)>0 $. For odd $ r $, $   f_1(r)=f_2(r)= 0 $.
		\item  $ f(r,s) \leqslant 0 $  for $ r  \geqslant d_x $ or $  s  \geqslant d_z $.
	\end{enumerate}
	Then \[ K \leqslant \frac{1}{2^n}
	\max_{ \substack{0 \leqslant i < d_x\\ 2|i} } \frac{f_1(i )}{\beta_{i}}
	\max_{ \substack{ 0 \leqslant j < d_z\\2|j} } \frac{f_2(j)}{\alpha_{j}}  . \]
\end{lem}

\begin{lem}
	Let $ A(x) = 2^{n-d+1}\prod_{r= d}^{n} \left(1- \frac{x}{r} \right)$ . Then
	\[ A(x) = \sum_{i = 0}^n \alpha_i P_i(x) \]
	where
	\begin{align}\label{eq:al_i}
	\alpha_i  =\alpha_i (d)  =  \binom{n-i}{d-1}  \Big / \binom{n}{n-d+1}.
	\end{align}
\end{lem}
\begin{proof}
	By definition, we have
	\[ A(x) =2^{n-d+1} \prod_{r= d}^{n} \left(1- \frac{x}{r} \right) = 2^{n-d+1}\binom{n-x}{n-d+1} \Big / \binom{n}{n-d+1}  .\]
It is known from the book \cite{MacWilliams1977} (Exercise 41) that
\[ 	 \sum_{r=0}^n \binom{n-r}{n-d+1} P_r(i) = 2^{d-1} \binom{n-i}{d-1}, \]
This implies that
	 \[  \alpha_i = 2^{-n} \sum_{r=0}^n  A(r) P_r(i)  =  \binom{n-i}{d-1}  \Big / \binom{n}{n-d+1}, \]
	 completing the proof of this lemma.
\end{proof}

\section{upper bounds}\label{sec:upper}

In this section, we extend the work of \cite{Ashikhmin1999} and derive asymptotic upper bounds on the size of an arbitrary asymmetric quantum code of given length and minimum distance. Precisely, the Singleton bound, the Hamming bound and the first linear programming bound are determined utilizing the key inequality presented in Section \ref{sec:Kraw}.
 \subsection{A Singleton-type bound}
\begin{thm}[Quantum Singleton Bound]
	Let $ Q $ be an $ [[n,k,d_z/d_x]] $ quantum code. Then  $ n \geqslant k+ d_x+d_z -2 $.
\end{thm}

 \begin{proof}
 Set
  $ \alpha_{i,j} = \alpha_i(d_z)\alpha_j(d_x) \geqslant 0 $ for $ 0 \leqslant i,j \leqslant n $, where $ \alpha_i(d) $ is defined in \eqref{eq:al_i}.
Let \begin{align*}
f(x,y) & = \sum_{i,j =0 }^n  \alpha_{i,j} P_i(y) P_j(x)\\
& = \sum_{i,j =0 }^n  \alpha_{i}(d_z) \alpha_{j}(d_x) P_i(y) P_j(x)  \\
&= 2^{2n - d_x-d_z +2} \prod_{r=d_x}^{n} \left(1- \frac{x}{r}\right)\prod_{s=d_z}^{n} \left(1- \frac{y}{s}\right).
\end{align*}
One may check that this polynomial verifies all conditions of Lemma \ref{lem:K}.
So  \begin{align*}
 K &\leqslant \frac{1}{2^n} \max_{0 \leqslant i < d_x ,  0 \leqslant j < d_z} \frac{f(i,j)}{\alpha_{i,j}}  \\
  &= 2^{n-d_x-d_z +2} \max_{0 \leqslant i <d_x } g(i; d_x) \max_{0 \leqslant j <d_z } g(j; d_z),
\end{align*}
where  \[ g(i; d) = \binom{n-i}{n-d+1} \Big/\binom{n-i}{d-1} . \]
For $ d \leqslant \frac{n}{2}+1 $, we have $  {g(i;d)}/{g(i+1;d)} \geqslant 1 $.
Therefore, we obtain
 \[ K  \leqslant 2^{n-d_x-d_z +2} g(0,d_x)g(0,d_z) = 2^{n-d_x-d_z +2}. \]
 Since $ K = 2^k $, we find $ n \geqslant k+ d_x+d_z -2 $.
  \end{proof}

 \subsection{A Hamming-type bound}
 Let $ \phi=\floor{\frac{d_x-1}{2}} $ and $ \theta= \floor{\frac{d_z-1}{2}}$. Define
 $ \alpha_{i}(d_z)=(P_{\theta}(i))^2 $, $ \beta_{j}(d_x)=(P_{\phi}(j))^2 $ and
 \begin{align}\label{eq:f(x,y)}
 	f(x,y)= \sum_{i,j =0 }^n  \alpha_{i}(d_z) \beta_{j}(d_x) P_i(y) P_j(x).
 \end{align}

 \begin{lem}[\cite{McEliece1977}, Eq A.19] \label{lem:prod}
 	Any product $ P_i(x)P_j(x) $ can be expressed as a linear combination of the
 	$ P_k(x) $ as follows:
 	\[ P_i(x)P_j(x) = \sum_{ k = 0 }^n \binom{n-k}{(i+j-k)/2} \binom{k}{(i-j+k)/2} P_k(x)  , \]
 	where a binomial coefficient with fractional or negative lower index is to be interpreted as zero.
 \end{lem}

Using the lemma, we get
\[ \alpha_{i}(d_z)= \sum_{ k = 0 }^n \binom{n-k}{\theta-k/2} \binom{k}{k/2} P_k(i) \]
and
\[ \beta_{j}(d_x)= \sum_{ k = 0 }^n \binom{n-k}{\phi-k/2} \binom{k}{k/2} P_k(j). \]
It then follows that
\begin{align*}
\sum_{ i = 0 }^n \alpha_{i}(d_z) P_i(y)
&= \sum_{ i = 0 }^n  \sum_{ k = 0 }^n \binom{n-k}{\theta-k/2} \binom{k}{k/2} P_k(i) P_i(y)\\
&=   \sum_{ k = 0 }^n \binom{n-k}{\theta-k/2} \binom{k}{k/2} \sum_{ i = 0 }^n  P_k(i) P_i(y) \\
&=2^n   \binom{n-y}{\theta-y/2} \binom{y}{y/2}  ,
\end{align*}
where in the last step we use \eqref{eq:orth}.
Similarly we have
\begin{align*}
\sum_{ j = 0 }^n \beta_{j}(d_x) P_j(x)
=2^n   \binom{n-x}{\phi-x/2} \binom{x}{x/2} ,
\end{align*}
Hence
\[f(x,y)= 2^{2n} \binom{n-x}{\phi-x/2} \binom{x}{x/2} \binom{n-y}{\theta-y/2} \binom{y}{y/2}. \]
For later use, we define the binary entropy function
$ H(x)=-x\log x -(1-x) \log (1-x) $
for $ 0 \leqslant x \leqslant 1 $.
Taking into account that
\begin{align}\label{eq:log(nk)}
\frac{1}{n} \log \binom{n}{k} = H\Bk{ \frac{k}{n}} +O \Bk{\frac{1}{n} }
\end{align}
and denoting $\xi= {x}/{n} $, $ \eta={y}/{n}  $, $ \tau = {\phi}/{n} $ and $\sigma={\theta}/{n} $, we get
\begin{align*}
&\frac{1}{n} \log \Biggr[\binom{n-x}{\phi-x/2} \binom{x}{x/2} \binom{n-y}{\theta-y/2} \binom{y}{y/2}\Biggr] \\
&=\frac{1}{n} \log \binom{n-x}{\phi-x/2} + \frac{1}{n} \log  \binom{x}{x/2} +\frac{1}{n} \log \binom{n-y}{\theta-y/2} +\frac{1}{n} \log \binom{y}{y/2}\\
&= (1-\xi) H\Bk{ \frac{\tau-\xi/2}{1-\xi}} +\xi H\Bk{ \frac{1}{2}}
+(1-\eta) H\Bk{ \frac{\sigma-\eta/2}{1-\eta}}+\eta  H\Bk{ \frac{1}{2}}+O \Bk{\frac{1}{n} }\\
&=  (1-\xi) H\Bk{ \frac{\tau-\xi/2}{1-\xi}} +\xi
+(1-\eta) H\Bk{ \frac{\sigma-\eta/2}{1-\eta}}+\eta+O \Bk{\frac{1}{n} }.
\end{align*}
This yields
\begin{align}\label{eq:logf}
 \frac{1}{n} \log f(x,y)= 2+\xi +\eta+ (1-\xi) H\Bk{ \frac{\tau-\xi/2}{1-\xi}}
 +(1-\eta) H\Bk{ \frac{\sigma-\eta/2}{1-\eta}}+O \Bk{\frac{1}{n} }.
\end{align}
To derive an estimate for  $  \alpha_{x,y}=\alpha_{y}(d_x) \alpha_{x}(d_z)  $, we need bounds on values of Krawtchouk polynomials. Recall that by \eqref{eq:bound1}
\begin{align}\label{eq:xiphi}
\gamma_{\phi}=\dfrac{r_{\phi}}{n}= \dfrac{1}{2}-\sqrt{\tau(1-\tau)}+ o(1),
\end{align}
and
\begin{align}\label{eq:etaphi}
\gamma_{\theta}=\dfrac{r_{\theta}}{n}= \dfrac{1}{2}-\sqrt{\sigma(1-\sigma)}+ o(1).
\end{align}
 We also recall the following equations, see \cite{Kalai1995}:
 \begin{align}\label{eq:logP}
 \frac{1}{n} \log P_{\phi}(x)= H(\tau)
 + \int_0^{\xi} \log \Bk{\dfrac{1-2 \tau +\sqrt{(1-2 \tau )^2-4z(1-z)}}{2(1-z)}} dz
 + O\Bk{ \frac{1}{n}}  ,
 \end{align}
for $ \xi <\gamma_{\phi} $ and
 \[\frac{1}{n} \log P_{\theta}(y)= H(\sigma)
+ \int_0^{\eta} \log \Bk{\dfrac{1-2 \sigma +\sqrt{(1-2 \sigma )^2-4z(1-z)}}{2(1-z)}} dz
+ O\Bk{ \frac{1}{n}},  \]
for $ \eta < \gamma_{\theta} $. Hence we obtain
\begin{align}\label{eq:alpha}
\frac{1}{n}\log \alpha_{x,y}
 &= \frac{1}{n}\log \Bk{(P_{\theta}(x))^2(P_{\phi}(y))^2} \nonumber\\
&= 2H(\tau)
+ 2\int_0^{\xi} \log \Bk{\dfrac{1-2 \tau +\sqrt{(1-2 \tau )^2-4z(1-z)}}{2(1-z)}} dz \nonumber \\
& \phantom{=}+2 H(\sigma)
+ 2\int_0^{\eta} \log \Bk{\dfrac{1-2 \sigma +\sqrt{(1-2 \sigma )^2-4z(1-z)}}{2(1-z)}} dz
+ O\Bk{ \frac{1}{n}}.
\end{align}

Now we are in a position to give the Hamming type bound.

\begin{thm}[Hamming-type Bound]\label{thm:Hambound}
	Let $ \tau=\floor{\frac{d_x-1}{2}}/n $ and $ \sigma=\floor{\frac{d_z-1}{2}}/n $.
	Define
	\begin{align*}
	 \Omega_{\tau}(\xi ):=
	   \xi
	 + (1-\xi) H\Bk{ \frac{\tau-\xi/2}{1-\xi}}
	 -2H(\tau)- 2\int_0^{\xi} \log \Bk{\dfrac{1-2 \tau +\sqrt{(1-2 \tau )^2-4z(1-z)}}{2(1-z)}} dz  .	
	\end{align*}
	Suppose that $ 2\tau <   \gamma_\phi $ and $ 2\sigma < \gamma_\theta  $ where $ \gamma_\phi$ and $ \gamma_\theta $ are given in \eqref{eq:xiphi} and \eqref{eq:etaphi}, respectively. Then for
	 an $ ((n,K,d_z/d_x ))  $ quantum code we have	
	\begin{align*}
	\frac{\log K}{n} \leqslant 1+ \max_{  0 \leqslant \xi \leqslant 2\tau    }  \Omega_{\tau}(\xi )
  + \max_{  0 \leqslant \eta \leqslant 2\sigma    }  \Omega_{\sigma}(\eta) + o(1).
	\end{align*}
\end{thm}
\begin{proof}
	It can be easily verified that the polynomial $ f(x,y) $ of \eqref{eq:f(x,y)} satisfies all the conditions of Lemma \ref{lem:K2}. So we get from \eqref{eq:logf}, \eqref{eq:alpha} and Lemma \ref{lem:K2} that
		\begin{align*}
	\frac{\log K}{n} & \leqslant -1 +   \max_{ \substack{ 0 \leqslant x < d_x,2|x \\
		0 \leqslant y < d_z,2|y}  }
	\Big\{\frac{1}{n} \log f(x,y) - \frac{1}{n}\log \alpha_{x,y}\Big\} 	\\
	& = -1 +   \max_{ \substack{ 0 \leqslant x \leqslant 2\phi ,2|x\\
			0 \leqslant y \leqslant 2 \theta,2|y}  }
	\Big\{\frac{1}{n} \log f(x,y) - \frac{1}{n}\log \alpha_{x,y}\Big\} 	\\
	&= 1 +  \max_{  0 \leqslant \xi \leqslant 2\tau    }
	\Big\{  \xi
	+ (1-\xi) H\Bk{ \frac{\tau-\xi/2}{1-\xi}}
	-2H(\tau)- 2\int_0^{\xi} \log \Bk{\dfrac{1-2 \tau +\sqrt{(1-2 \tau )^2-4z(1-z)}}{2(1-z)}} dz \Big\}
	\\
	&\phantom{=} + \max_{     0 \leqslant \eta \leqslant 2\sigma  }
	\Big\{  \eta
	+(1-\eta) H\Bk{ \frac{\sigma-\eta/2}{1-\eta}}
	- 2 H(\sigma)  - 2\int_0^{\eta} \log \Bk{\dfrac{1-2 \sigma +\sqrt{(1-2 \sigma )^2-4z(1-z)}}{2(1-z)}} dz 	
	\Big\} + o(1),
	 \end{align*}
	where in the third step we use the fact that $ \xi =  x/ n \leqslant 2\phi/n = 2 \tau <\gamma_{\phi} $ and $ \eta =  y/ n \leqslant 2\theta /n = 2 \sigma <\gamma_{\theta} $.
	This completes the proof.
\end{proof}

 Matlab shows that the function in Theorem \ref{thm:Hambound} achieves its maximum at $ \xi=0 $ and $ \eta=0 $ for any $ 2\tau<  \gamma_{\phi} $ and $ 2\sigma< \gamma_{\theta} $. A straightforward calculation gives that
 $ h(x)=2x+\sqrt{ x  (1- x )} -1/2$ is a monotone increasing function in $  x $ if $ 0 \leqslant  x \leqslant 1/2 $. Denote $ \delta_x= d_x/n $ and $ \delta_z=d_z/n $.
 Assume that $ 0   \leqslant  \delta_x \leqslant 1/5 $. Then we have $   \tau =\floor{\frac{d_x-1}{2}}/n < d_x/2n= \delta_x /2 $.
 Let $ \Delta= h'(1/10 )(\tau-1/10)<0 $ where $ h' $ denotes the first derivative of $ h $. Then we have
 \begin{align*}
   \frac{h(1/10   )- h(\tau)}{1/10  -\tau}=h'(\vartheta) > h'(1/10),
  \end{align*}
   where $ \tau < \vartheta < 1/10  $. Since $   h(1/10   )  = 0  $, then
   \begin{align*}
   h(\tau)  <   h'(1/10 )(\tau-1/10)= \Delta.
   \end{align*}
  It follows that  $ 2 \tau < 1/2 - \sqrt{ \tau  (1-\tau  )} + \Delta \leqslant   \gamma_{\phi} $ by noting that $ \Delta < o(1) $ for $ n $ sufficiently large. Therefore
  we conclude that if $ 0 \leqslant \delta_x \leqslant 1/5 $ and $ 0 \leqslant \delta_z \leqslant 1/5 $ then $ 2 \tau < \gamma_{\phi} $ and $ 2 \sigma < \gamma_{\theta} $. This means the conventional Hamming bound is valid when $ 0 \leqslant \delta_x \leqslant 1/5 $ and $ 0 \leqslant \delta_z \leqslant 1/5 $.
\begin{cor}
	The conventional Hamming bound is valid for quantum codes, i.e., if $ Q $ is
	an $ ((n, K, d_z/ d_x)) $ quantum code then
		\begin{align*}
	\frac{\log K}{n} \leqslant 1-H\Bk{\frac{\delta_x}{2}}- H\Bk{ \frac{\delta_z}{2}} + o(1),
	\end{align*}
	where $ \delta_x= d_x/n $ and $ \delta_z=d_z/n $ stand for the relative distances of the code, $ 0 \leqslant \delta_x \leqslant 1/5 $ and $ 0 \leqslant \delta_z \leqslant 1/5 $.
\end{cor}
\begin{proof}
	Taking  $ \xi=0 $ and $ \eta=0 $ in Theorem \ref{thm:Hambound} yields that
	\begin{align*}
	\frac{\log K}{n} & \leqslant 1-H(\tau)-H(\sigma) + o(1)\\
	&= 1-H\Bk{\frac{\delta_x}{2}}- H\Bk{ \frac{\delta_z}{2}} + o(1),
	\end{align*}
 where in the last step we use the fact $ H(\tau ) = H(\delta_x/2)+o(1)$ and
$  H(\sigma)  =   H(\delta_z/2)+o(1) $.	
\end{proof}

\subsection{The First Linear Programming Bound}

To get the first linear programming bound for an $ ((n,K, d_z/d_x)) $ code, one has to choose integers $ s  $ and $ t $ such that
\begin{align*}
&\dfrac{t}{n}=\frac{1}{2}-\sqrt{\delta_x(1-\delta_x)} +o(1),\\
&\dfrac{s}{n}=\frac{1}{2}-\sqrt{\delta_z(1-\delta_z)} +o(1),
\end{align*}
where $ \delta_x=d_x/n $, $\delta_z=d_z/n   $. Then we choose integers $ a $ and $ b $ such that $ r_{t+1} <a <r_t$, $ r_{s+1} <b <r_s$, $  {P_t(a)}/{P_{t+1}(a)}=-1 $ and $ {P_s(b)}/{P_{s+1}(b)}=-1 $.
Define
 \begin{align}\label{eq:f(x,y)2}
 f(x,y)=F(x) F(y),
 \end{align}
 where
  \begin{align*}
 &F(x)= \dfrac{1}{a-x} \{ P_{t+1}(x)P_t(a) - P_t(x)P_{t+1}(a) \}^2,\\
 &G(y)= \dfrac{1}{b-y} \{ P_{s+1}(y)P_s(b) - P_s(y)P_{s+1}(b) \}^2.
 \end{align*}
 This polynomial \eqref{eq:f(x,y)2} will yield the first linear programming bound for classical codes over $ \mathbb{F}_4 $ \cite{Aaltonen1979,Levenshtein1995,Lai2018}.
 By the Christoffel-Darboux formula \eqref{eq:CDeq}
 \begin{align*}
 F(x)& =\dfrac{2}{t+1} \binom{n}{t}
  \{ P_{t+1}(x)P_t(a)-P_t(x)P_{t+1}(a) \}
  \sum_{ i = 0 }^t \dfrac{P_i(x)P_i(a)}{\binom{n}{i}}\\
   & =\dfrac{2}{t+1} \binom{n}{t}
  P_t(a)
   \sum_{ i = 0 }^t \dfrac{ P_i(a)}{\binom{n}{i}} \{ P_{t+1}(x)P_i(x)+ P_t(x)P_i(x) \}.
 \end{align*}
 It follows from Lemma \ref{lem:prod} that
  \begin{align*}
 F(x) & =\dfrac{2}{t+1} \binom{n}{t}
     P_t(a)
    \sum_{ i = 0 }^t \dfrac{ P_i(a)}{\binom{n}{i}}\\
   &\phantom{=}  \times \Bigg\{ \sum_{ j = 0 }^n P_j(x) \binom{n-j}{(t+1+i-j)/2} \binom{j}{(t+1-i+j)/2}
    	+\sum_{ j = 0 }^n P_j(x) \binom{n-j}{(t+i-j)/2} \binom{j}{(t-i+j)/2}
    	 \Bigg\}\\
    &= \sum_{ j = 0 }^n P_j(x) \dfrac{2}{t+1} \binom{n}{t}
    P_t(a)
    \sum_{ i = 0 }^t \dfrac{ P_i(a)}{\binom{n}{i}} \\
     &\phantom{=}  \times \Bigg\{   \binom{n-j}{(t+1+i-j)/2} \binom{j}{(t+1-i+j)/2}
    +  \binom{n-j}{(t+i-j)/2} \binom{j}{(t-i+j)/2}
    \Bigg\}\\
    &= \sum_{ j = 0 }^n P_j(x) F_j ,
   \end{align*}
  where we denote
  \begin{align*}
  F_j=\dfrac{2}{t+1} \binom{n}{t}
  P_t(a)
  \sum_{ i = 0 }^t \dfrac{ P_i(a)}{\binom{n}{i}}
    \times \Bigg\{   \binom{n-j}{(t+1+i-j)/2} \binom{j}{(t+1-i+j)/2}
  +  \binom{n-j}{(t+i-j)/2} \binom{j}{(t-i+j)/2}\Bigg\} .
  \end{align*}
 Taking $ j=x $ and estimating $ F_x $ by the term with $ i=t $, we obtain
   \begin{align*}
   F_x & \geqslant
   \dfrac{2}{t+1} \binom{n}{t}
      \dfrac{ P_t(a)^2}{\binom{n}{t}}
    \binom{n-x}{t-x/2} \binom{x}{x/2}\\
    &= \dfrac{2}{t+1}
      P_t(a)^2
    \binom{n-x}{t-x/2} \binom{x}{x/2} .
   \end{align*}
   Denote $ \xi=x/n $, $ \tau=t/n $. Then, similarly to the derivation of the Hamming bound, we have
   \begin{align}\label{eq:binom}
   \dfrac{1}{n} \log \Big\{
   \binom{n-x}{t-x/2} \binom{x}{x/2} \Big\}
   =  (1-\xi) H\Bk{ \frac{\tau-\xi/2 }{ 1-\xi}} + \xi +O\Bk{\frac{1}{n} }.
   \end{align}
  Then, using \eqref{eq:sym}, we get
  \begin{align*}
  \dfrac{F(x)}{F_x} & \leqslant
  \dfrac{(t+1)P_t(a)^2 \{P_{t+1}(x)+P_t(x)\}^2}{2 (a-x) P_t(a)^2 \binom{n-x}{t-x/2} \binom{x}{x/2}  } \\
  &=   \dfrac{(t+1)\Bigg\{ \dfrac{\binom{n}{t+1} P_x(t+1)}{\binom{n}{x}}
  	 + \dfrac{\binom{n}{t} P_x(t)}{\binom{n}{x}}
  	 \Bigg\}^2}
  {2 (a-x)  \binom{n-x}{t-x/2} \binom{x}{x/2}  } \\
  &=   \dfrac{(t+1)\binom{n}{t}^2 \Big\{ \dfrac{n-t}{t+1} P_x(t+1)
	+ P_x(t)	\Big\}^2}
  {2 (a-x) \binom{n}{x}^2  \binom{n-x}{t-x/2} \binom{x}{x/2}  } .
  \end{align*}
It then follows from \eqref{eq:sym} and \eqref{eq:ratio} that
  \begin{align*}
\dfrac{F(x)}{F_x} & \leqslant
   \dfrac{ P_t(x)^2 \Big\{ \dfrac{n-2x+\sqrt{(n-2x)^2-4t(n-t)}}{2} +t+1 	\Big\}^2}
{2 (a-x)(t+1)   \binom{n-x}{t-x/2} \binom{x}{x/2}  } .
\end{align*}
Taking logarithm on both sides and dividing by $ n $, we obtain from \eqref{eq:log(nk)}, \eqref{eq:logP} and \eqref{eq:binom} that
  \begin{align}\label{eq:F(x)}
\dfrac{1}{n}\log \dfrac{F(x)}{F_x} & \leqslant
2H(\tau) +2 \int_0^{\xi} \log \Bk{\dfrac{1-2 \tau +\sqrt{(1-2 \tau )^2-4z(1-z)}}{2(1-z)}} dz - (1-\xi) H\Bk{ \frac{\tau-\xi/2 }{ 1-\xi}} - \xi +O\Bk{\frac{1}{n} },
\end{align}
for $ \xi <\gamma_{\phi} $, where $ \gamma_{\phi}  $ is given in \eqref{eq:xiphi}. By a similar argument as above, we have
  \begin{align*}
G(y)
= \sum_{ i = 0 }^n P_i(y)G_i,
\end{align*}
 where
 \begin{align*}
 G_i=\dfrac{2}{s+1} \binom{n}{s}
P_s(b)
\sum_{ j = 0 }^s \dfrac{ P_j(b)}{\binom{n}{j}}    \times \Bigg\{   \binom{n-i}{(s+1+j-i)/2} \binom{i}{(s+1-j+i)/2}
+  \binom{n-i}{(s+j-i)/2} \binom{i}{(s-j+i)/2}
\Bigg\},
\end{align*}
and consequently we obtain
  \begin{align}\label{eq:G(y)}
\dfrac{1}{n}\log \dfrac{G(y)}{G_y} & \leqslant
2H(\sigma) +2 \int_0^{\eta} \log \Bk{\dfrac{1-2 \sigma +\sqrt{(1-2 \sigma )^2-4z(1-z)}}{2(1-z)}} dz - (1-\eta) H\Bk{ \frac{\sigma-\eta/2 }{1-\eta}} - \eta +O\Bk{\frac{1}{n} },
\end{align}
for $ \eta < \gamma_{\theta} $, where $ \sigma=s/n $, $ \eta=y/n $  and $ \gamma_{\theta} $ is given in \eqref{eq:etaphi}.

 With the above preparation, we can get the following theorem.
\begin{thm}
	Let $ \delta_x=d_x/n $, $ \delta_z=d_z/n $,
	$ \tau =\frac{1}{2}-\sqrt{\delta_x(1-\delta_x)} $ and
	$ \sigma=\frac{1}{2}-\sqrt{\delta_z(1-\delta_z)} $. Let
	\begin{align*}
		 \Gamma_{\tau} (\xi) :=2H(\tau) +2 \int_0^{\xi} \log \Bk{\dfrac{1-2 \tau +\sqrt{(1-2 \tau )^2-4z(1-z)}}{2(1-z)}} dz - (1-\xi) H\Bk{ \frac{\tau-\xi/2 }{ 1-\xi}} - \xi.
	\end{align*}	
	Suppose that $ \delta_x < \gamma_{\phi}   $ and $ \delta_z<\gamma_{\theta}  $. Then for an $ ((n,K, d_z/d_x)) $ code, we have
	\begin{align*}
	\frac{\log K}{n}
	&\leqslant  -1 + \max_{  0 \leqslant \xi < \delta_x} \Gamma_{\tau} (\xi)  + \max_{  0 \leqslant \eta < \delta_z}  \Gamma_{\sigma}(\eta) +o(1).
	\end{align*}
	
\end{thm}

\begin{proof}
	One verifies that the polynomial $ f(x,y) $ of \eqref{eq:f(x,y)2} satisfies all the conditions of Lemma \ref{lem:K}.
	Therefore applying Lemma \ref{lem:K}, \eqref{eq:F(x)} and \eqref{eq:G(y)} gives that
	\begin{align*}
	\frac{\log K}{n}
	&\leqslant \max_{   0 \leqslant x < d_x }  \Big\{ \dfrac{1}{n}\log \dfrac{F(x)}{F_x}
	\Big\}+ \max_{     0 \leqslant y < d_z  }  \Big\{ \dfrac{1}{n}\log \dfrac{G(y)}{G_y}
	\Big\}  -1+O\Bk{\frac{1}{n} } \\
	&=-1 + \max_{  0 \leqslant \xi < \delta_x} \Big\{ 2H(\tau) +2 \int_0^{\xi} \log \Bk{\dfrac{1-2 \tau +\sqrt{(1-2 \tau )^2-4z(1-z)}}{2(1-z)}} dz - (1-\xi) H\Bk{ \frac{\tau-\xi/2 }{ 1-\xi}} - \xi \Big\}\\
	&\phantom{=}+ \max_{  0 \leqslant \eta < \delta_z} \Big\{2H(\sigma) +2 \int_0^{\eta} \log \Bk{\dfrac{1-2 \sigma +\sqrt{(1-2 \sigma )^2-4z(1-z)}}{2(1-z)}} dz - (1-\eta) H\Bk{ \frac{\sigma-\eta/2 }{1-\eta}} - \eta \Big\} +o(1),
	\end{align*}
	where in the second step we take $ \xi=x/n $ and $ \eta=y/n $. This completes the proof.
\end{proof}
Computations with Matlab show that this function achieves its minimum at $ \xi=0 $ and $ \eta=0 $ for any $ \delta_x  \leqslant 0.1865 $ and $ \delta_z \leqslant 0.1865 $.

\begin{cor}[The First Linear Programming Bound]
	If  $ 0 \leqslant \delta_x  \leqslant 0.1865$ and
	$ 0 \leqslant \delta_z \leqslant 0.1865 $ then the conventional linear programming bound is valid for quantum codes, i.e., if $ Q $ is an $ ((n,K,d_z/d_x)) $ quantum code then
	\begin{align*}
	\frac{\log K}{n}
	\leqslant   H\Bk{\frac{1}{2}-\sqrt{\delta_x(1-\delta_x)}} + H\Bk{\frac{1}{2}-\sqrt{\delta_z(1-\delta_z)}}  -1  +o(1).
	\end{align*}
\end{cor}

   A straightforward computation gives that when $  \delta_x = \delta_z = 0.1865$, $   {\log K}/{n}\approx 0.0028 $. So for all $  0.0028\leqslant {\log K}/{n}\leqslant 1 $, the conventional first linear programming bound is valid.

\section*{Acknowledgment}
 	The work is partially supported by the NSFC (11701317,11531007,11431015,61472457,11571380) and Tsinghua University startup fund. Prof. Stephen S.-T. Yau is grateful to National Center for Theoretical Sciences (NCTS) for providing excellent research environment while part of this research was done. This work is also partially supported by China Postdoctoral Science Foundation Funded Project (2017M611801), Jiangsu Planned Projects for Postdoctoral Research Funds (1701104C) and Guangzhou Science and Technology Program (201607010144).

  ------------------------------------------

\ifCLASSOPTIONcaptionsoff
  \newpage
\fi



%

 

\end{document}